\documentclass[preprint,authoryear,3p,12pt]{elsarticle}

\usepackage{amssymb}
\usepackage{amsthm}

\usepackage{mathrsfs}
\usepackage{float}
\usepackage{eufrak}

\usepackage{graphics}
\usepackage{amsfonts}
\usepackage{mathrsfs}
\usepackage{amscd}
\usepackage{color}
\usepackage{url}
\usepackage[plainpages=false,pdfpagelabels=true,colorlinks=true,citecolor=blue,hypertexnames=false]{hyperref}
\usepackage[ruled,vlined]{algorithm2e}

\journal{}

\newtheorem{theorem}{Theorem}[section]
\newtheorem{prop}[theorem]{Proposition}

\newtheorem{lemma}[theorem]{Lemma}
\newtheorem{remark}[theorem]{Remark}
\newtheorem{define}[theorem]{Definition}
\newtheorem{example}[theorem]{Example}

\def\Q{{\mathbb{Q}}}

\def\k{{\rm K}}

\def\lpp{{\rm lpp}}
\def\lc{{\rm lc}}
\def\f{{\bf f}}
\def\e{{\bf e}}

\def\u{{\bf u}}
\def\v{{\bf v}}
\def\w{{\bf w}}
\def\fu{{f^{[\u]}}}
\def\gv{{g^{[\v]}}}
\def\hw{{h^{[\w]}}}
\def\bru{\bar{\u}}
\def\brv{\bar{\v}}
\def\brf{\bar{f}}
\def\brg{\bar{g}}
\def\brfu{{\brf^{[\bru]}}}
\def\brgv{{\brg^{[\brv]}}}
\def\lcm{{\rm lcm}}

\def\deg{{\rm deg}}
\def\max{{\rm max}}

\def\x{{x_1,\cdots,x_n}}

\def\lla{{\longleftarrow}}

\newcommand{\vh}{\vspace*{8pt}}
\newcommand{\spc}{\hspace*{15pt}}

\newcommand{\comment}[1]{}
\newcommand{\ignore}[1]{}
\newcommand{\gr}{Gr\"obner\,}

\begin{document}

\begin{frontmatter}



\title{A Generalized Criterion for Signature-based Algorithms to Compute Gr\"obner Bases \tnoteref{labeltitle}}


\author{Yao Sun, Dingkang Wang \label{label3}}

\tnotetext[labeltitle]
{This paper is a substantially expanded version of the paper entitled ``A Generalized Criterion for  Signature Related \gr Basis Algorithms'', which was presented at ISSAC 2011 \citep{SunWang11}.}

\fntext[label3]{The authors are supported by NKBRPC 2011CB302400, NSFC 10971217 and 60821002/F02.}

\address{Key Laboratory of Mathematics Mechanization\\ Academy of Mathematics and Systems Science, Chinese Academy of Sciences, Beijing 100190, China}

\ead{sunyao@amss.ac.cn, dwang@mmrc.iss.ac.cn}

\begin{abstract}
A generalized criterion for signature-based algorithms to compute \gr bases is proposed in this paper. This criterion is named by ``generalized criterion", because it can be specialized to almost all existing criteria for signature-based algorithms which include the famous F5 algorithm, F5C, extended F5, G$^2$V and the GVW algorithm. The main purpose of current paper is to study in theory which kind of criteria is correct in signature-based algorithms and provide a generalized method to develop new criteria. For this purpose, by studying some key facts and observations of signature-based algorithms, a generalized criterion is proposed. The generalized criterion only relies on a partial order defined on a set of polynomials. When specializing the partial order to appropriate specific orders, the generalized criterion can specialize to almost all existing criteria of signature-based algorithms. For {\em admissible} partial orders, a proof is  presented for the correctness of the algorithm that is based on this generalized criterion. And the partial orders implied by the criteria of F5 and GVW are also shown to be admissible. More importantly, the generalized criterion provides an effective method to check whether a new criterion is correct as well as to develop new criteria for signature-based algorithms.
\end{abstract}


\end{frontmatter}



\smallskip

\section{Introduction}

\gr basis was first proposed by Buchberger in 1965. Since then, many important improvements have been made to speed up the algorithms for computing \gr basis \citep{Buchberger79, Buchberger85, GebMol86, Gio91, Mora92, Fau99, Fau02}. One important improvement is that Lazard pointed out the connection between a \gr basis and linear algebra \citep{Lazard83}. This idea is also implemented as XL type algorithms by Courtois et al. \citep{Courtois00} and Ding et al. \citep{Ding08}. Up to now, F5 is one of the most efficient algorithms for computing \gr basis. The notion of ``signatures" for polynomials was also introduced by Faug\`ere in \citep{Fau02}. Since F5 was proposed in 2002, it has been widely investigated and several variants of F5  have   been presented, including the F5C algorithm \citep{Eder09} and F5 with extended criteria \citep{Ars09}. Proofs and other extensions of F5 are also investigated in \citep{Stegers05, Eder08, Albrecht10, Arri10, SunWang10a, SunWang10b, Zobnin10}. Recently, Gao et al. proposed an incremental signature-based algorithm G$^2$V to compute \gr basis in \citep{Gao09}, and presented an extended version GVW in \citep{Gao10b}. The framework of signature-based algorithms was studied in \citep{Eder11}.

The common characteristics of F5, F5C, extended F5, G$^2$V and GVW are (1) each polynomial has been assigned a {\em signature}, and (2) both the criteria and the  reduction process  depend on the signatures of polynomials. So all these algorithms are signature-based algorithms. The only difference among the algorithms is that their criteria are different.

By studying these criteria carefully, we find a key fact in signature-based algorithms, and then some observations are motivated. One key observation is that if two polynomials have the same signature, then at most one of them is necessary to be reduced. The reason is that reducing two polynomials that have the same signature, could create the same leading power product if some extra conditions hold. With this insight, we use a partial order to help choose one polynomial that is not to be reduced. Then a generalized criterion for signature-based algorithms is proposed based on this partial order. By using appropriate partial orders, the generalized criterion can be specialized to almost all existing criteria of signature-based algorithms.

Unfortunately, not all partial orders can make the generalized criterion correct. We proved   that the generalized criterion is correct if the partial order is {\em admissible}.
Moveover, we show that the partial orders implied by F5 and GVW's criteria are both admissible, so the proof in this paper is also valid for the correctness of F5 and GVW.

The significance of the generalized criterion is to show which kind of criteria is correct for signature-based algorithms and provide a generalized method to check or even develop new criteria. Specifically, when a new criterion is presented, if it can be specified from the generalized criterion by using an admissible partial order, then this new criterion is definitely correct. It is also possible for us to develop some new criteria by using new admissible partial orders in the generalized criterion. From the proof in this paper, we know that any admissible partial order can develop a new criterion for signature-based algorithms in theory, but not all of these criteria can reject as many critical pairs as possible. Therefore, we believe that if the admissible partial order is in fact a total order, then almost all useless computations can be avoided. The proof for the claim will be included in our future works.

The paper is organized as follows. We present our main ideas of the generalized criterion in Section \ref{sec_mainideas}. Section \ref{sec_criterion} gives the generalized criterion and shows how this generalized criterion is used. Section \ref{sec_specializations} details how the generalized criterion specializes to F5 and GVW's criteria. We prove the correctness of the generalized criterion in Section \ref{sec_proof}. A new criterion is developed in Section \ref{sec_newcri}. Concluding remarks follow in Section \ref{sec_conclusion}.



\section{Main ideas} \label{sec_mainideas}

\subsection{Problem} \label{subsec_problem}

Let $R:=\k[\x]$ be a polynomial  ring over a field $\k$ with $n$ variables. Suppose $\{f_1, \cdots, f_m\}$  is a finite subset of  $ R$.  We want to compute a \gr basis for the ideal
$$I:=\langle f_1, \cdots, f_m\rangle=\{p_1f_1+\cdots+p_mf_m \mid p_1,\cdots,p_m \in R\} $$
with respect to some term order on $R$.

Fix a term order $\prec_1$ on $R$. We define the {\em leading power product} and {\em leading coefficient} of a polynomial $f\in R$ to be $\lpp(f)$ and $\lc(f)$ in general way. For example, let $f := 2x^2y+3z \in \Q[x, y, z]$ where $\Q$ is the rational field. Then $\lpp(f)=x^2y$ and $\lc(f) = 2$.

As we know, a set $G\subset I$ is a \gr basis for $I$, if and only if $$\langle \lpp(G)\rangle = \langle \lpp(I)\rangle.$$ That is, a \gr basis should contain all the leading power product information of $I$. So in order to compute a \gr basis for $I$, all existing algorithms start with a set of known generators, and then a \gr basis can be obtained by expanding these known generators constantly with polynomials having new leading power products. To get the polynomials that have new leading power products, the only way is to {\em reduce} polynomials in $I$. However, if a polynomial is reduced to $0$, then this reduction is redundant, since no new leading power product appears. In this case, criteria for \gr basis algorithms are created, and {\em all criteria aim to avoid computations that reduce polynomials to $0$}.

Now we should answer an important question: given $f\in I$ and $G\subset I$, {\bf how can we {\em predict} the reducing result of $f$ by $G$ without really reducing it?}

{\em Signature-based algorithms} give a good solution to this question, and we notice that their common methods are based on {\bf ordering the polynomials in $I$ according their signatures in order to get a beautiful property}. This beautiful property is a key fact in signature-based algorithms, and it will be presented in Subsection \ref{subsec_keyfact}. First, let us see what is the signature of a polynomial in $I$.

\subsection{Signature}

We will use the following simple example to help illustrate some notions in this subsection, and these notions can be extended to general case easily.

\begin{example} \label{exa_mainidea}
Let $I:=\langle f_1, f_2, f_3 \rangle$ be an ideal in the polynomial ring $R=\Q[x, y, z]$, where $f_1 = yz-x$, $f_2 = xz-y$, $f_3 = xy-z$. The term order $\prec_1$ is the Degree Reverse Lex order with $(x \succ y \succ z)$.
\end{example}

For a polynomial $f=y^2-z^2 \in I$, since $\{f_1, f_2, f_3\}$ is a set of generators of $I$, the polynomial $f$ has a representation w.r.t. $f_1, f_2, f_3$: $$f=0\cdot f_1 -y\cdot f_2 + z\cdot f_3=(0, -y, z)\cdot (f_1, f_2, f_3),$$ where ``$\cdot$" is the inner product of two vectors.

The vector $(f_1, f_2, f_3)$ is fixed to the ideal $I$, so the polynomial $f$ is determined by the vector $(0, -y, z)$. Let $\u := (0, -y, z)\in \Q[x,y,z]^3$. Then the vector $\u$ can be regard as an {\bf ID} of $f$. Note that ID of $f=y^2-z^2$ is not unique. For example, $\u' = (xz-y, -yz+x-y, z)\in \Q[x,y,z]^3$ is also an ID of this $f$.

In general case, for any $f\in I$, there always exists $\u = (p_1, p_2, p_3) \in R^3=\Q[x,y,z]^3$, such that $$f  = \u \cdot (f_1, f_2, f_3) = p_1\cdot f_1 + p_2\cdot f_2 + p_3\cdot f_3.$$
That is, any polynomial in $I$ has at least one ID.  To express the relation between $f$ and $\u$, we use the notation $f^{[\u]}$, which means $f=\u \cdot (f_1, f_2, f_3)$. \footnote{An equivalent notation $(\u, f)$ is used in \citep{SunWang11}. Now we prefer $f^{[\u]}$ to $(\u, f)$, since the notation $f^{[\u]}$ indicates $\u$ is only an auxiliary value to $f$.} For convenience, we also call $f^{[\u]}$ to be a polynomial in $I$. For example, ${(y^2-z^2)}^{[-y\e_2 + z\e_3]}$ and ${(y^2-z^2)}^{[(xz-y)\e_1 + (-yz+x-y)\e_2 + z\e_3]}$ are two polynomials in $I$, and {\em we treat ${(y^2-z^2)}^{[-y\e_2 + z\e_3]}$ and ${(y^2-z^2)}^{[(xz-y)\e_1 + (-yz+x-y)\e_2 + z\e_3]}$ as different polynomials in this paper}, i.e. $\fu=f'^{[\u']}$ if and only if $f=f'$ and $\u = \u'$.

The computations on $f^{[\u]}$ can be defined naturally. Let $f^{[\u]}$ and $g^{[\v]}$ be two polynomials such that $f=\u \cdot (f_1, f_2, f_3)$ and $g = \v \cdot (f_1, f_2, f_3)$, $c$ be a constant in $\Q$ and $t$ be a power product in $R$. Then
\begin{enumerate}
\item $\fu + \gv = (f+g)^{[\u+\v]}$.

\item $ct (f^{[\u]}) = (ctf)^{[ct\u]}$.
\end{enumerate}
Clearly, the operations are well defined, i.e. $f+g = (\u+\v) \cdot (f_1, f_2, f_3)$ and $ctf = (ct\u) \cdot (f_1, f_2, f_3)$.

Since $\u$ is a vector in the free module $R^3$, we consider a term order $\prec_2$ on $R^3$. The term order $\prec_2$ can be any admissible term order. In this example, we use the term order introduced in F5, i.e. $$x^\alpha\e_i \prec_2 x^\beta\e_j    \mbox{  iff  } \left\{\begin{array}{l} i > j, \\ \mbox{ or } \\ i = j \mbox{ and } x^\alpha \prec_1 x^\beta, \end{array}\right.$$ where $\e_1=(1, 0, 0)$, $\e_2 = (0, 1, 0)$ and $\e_3=(0, 0, 1)$. When the term order on $R^3$ is fixed, we can define the {\em leading power product} and {\em leading coefficient} of $\u=(p_1, p_2, p_3)=p_1\e_1+p_2\e_2+p_3\e_3 \in R^3$ to be $\lpp(\u)$ and $\lc(\u)$ similarly. More related definitions on ``module" can be found in Chapter 5 of \citep{CLO04}.

Then for a polynomial $f^{[\u]}$ where $f=\u\cdot (f_1, f_2, f_3)$, we define $\lpp(\u)$ to be the {\bf signature} of $f^{[\u]}$. For example, the signature of ${(y^2-z^2)}^{[-y\e_2 + z\e_3]}$ is $\lpp(-y\e_2+z\e_3) = y\e_2$. Original definition of signature is introduced by Faug\`ere in \citep{Fau02}, and recently, Gao et al. give a generalized definition of signature in \citep{Gao10b}. In this paper, we use the generalized definition given by Gao et al.

With signatures, we can then compare polynomials in $I$ w.r.t. their signatures. That is, we say $f^{[\u]}$ is {\em bigger} than $g^{[\v]}$, if $f^{[\u]}$ has bigger signature than $g^{[\v]}$, i.e. $\lpp(\u) \succ_2 \lpp(\v)$. Now we actually set up an {\em ordering} on the polynomials in $I$. Moreover, if we deal with the polynomials according to this ordering, we will have a very beautiful property, which is the key fact in next subsection.



\subsection{Key fact} \label{subsec_keyfact}

For a general ideal $I=\langle f_1, \cdots, f_m\rangle \subset R$, we find the following key fact.

\vh
\noindent
{\bf Key Fact:} {\em
Let $f^{[\u]}$, $g^{[\v]}$ be two polynomials and $G$ be a subset of $I$. Suppose $f^{[\u]}$ and $g^{[\v]}$ are reduced to ${f'}^{[\u']}$ and ${g'}^{[\v']}$ by $G$ respectively. Then $f'$ and $g'$ have the same leading power product, i.e. $\lpp({f'})=\lpp({g'})$, if {\bf\em $f^{[\u]}$ and $g^{[\v]}$ have the same signature}, i.e. $\lpp(\u) = \lpp(\v)$, and two extra conditions hold.
}\vh

Briefly, Key Fact means that {\em reducing $f$ and $g$ could create the same leading power product if  $f^{[\u]}$ and $g^{[\v]}$ have the same signature}. This fact is very interesting and important, from which we can predict the reducing result of polynomials without really reducing them, and hence we can answer the important question proposed in Subsection \ref{subsec_problem}.

Next, let us see the two extra conditions. We emphasize the first condition.

\vh\noindent
{\bf Condition 1:}
{\em A {\bf\em one-side-reduction}, which is defined below, must be used in Key Fact.}

\begin{define} \label{defn_reduce}
We say $f^{[\u]}$ is {\bf\em reducible} by $h^{[\w]}\in G$, only if
\begin{enumerate}
\item $\lpp(h)$ divides $\lpp(f)$, and

\item $t(\hw)$'s signature $\prec_2$ $f^{[\u]}$'s signature, i.e. $\lpp(t\w) \prec_2 \lpp(\u)$ where $t = \lpp(f)/\lpp(h)$.
\end{enumerate}

If $f^{[\u]}$ is reducible by $h^{[\w]}\in G$, then $f^{[\u]}\longmapsto_G f^{[\u]}-ct(h^{[\w]})$ is called a {\bf\em one-step-reduction} by $G$ where $c=\lc(f)/\lc(h)$ and $t = \lpp(f)/\lpp(h)$.

We say {\bf\em $f^{[\u]}$ is reduced to ${f'}^{[\u']}$ by $G$}, if ${f'}^{[\u']}$ is obtained by several one-step-reductions from $f^{[\u]}$, and ${f'}^{[\u']}$ is not reducible by $G$.
\end{define}

In simple words, this one-side-reduction indicates {\em $f^{[\u]}$ can only be reduced by polynomials having smaller signatures}. In Example \ref{exa_mainidea}, $(xyz-y^2)^{[y\e_2]}$ is {\em reducible} by $(xy-z)^{[\e_3]}$ but {\em not reducible} by $(yz-x)^{[\e_1]}$. The reason comes from the constraint of signatures.

Note that for the result $f^{[\u]}-ct(h^{[\w]})=(f-cth)^{[\u-ct\w]}$ of the one-step-reduction, we still have $(\u-ct\w) \cdot (f_1, \cdots, f_m) = \u \cdot (f_1, \cdots, f_m) - ct \w \cdot (f_1, \cdots, f_m) = f - cth$. So for ${f'}^{[\u']}$, the equation $f'=\u'\cdot (f_1, \cdots, f_m)$ also holds. Moreover, if $f^{[\u]}$ is reduced to ${f'}^{[\u']}$, then $f^{[\u]}$ and ${f'}^{[\u']}$ must have the same signature, i.e. $\lpp(\u) = \lpp(\u')$.

Now we can see how the {\em ordering} on the polynomials is used in Key Fact. That is, if  $f^{[\u]}$ and $g^{[\v]}$ have the same signature, and only the polynomials having smaller signatures are used to reduce $f^{[\u]}$ and $g^{[\v]}$, then the reducing results $f'$ and $g'$ could have the same leading power product.

Therefore, this one-side-reduction is a necessary condition to the key fact. We notice that {\em all existing signature-based algorithms are using this kind of one-side-reduction.}

\vh\noindent
{\bf Condition 2:}
{\em For any $\brf^{[\bru]} \in I$ with $\brf^{[\bru]}$'s signature $\prec_2$ $\fu$'s signature, i.e. $\lpp(\bru) \prec_2 \lpp(\u)$, there always exists $h^{[\w]}\in G$ such that
\begin{enumerate}
\item $\lpp(h)$ divides $\lpp(\bru)$, and

\item $t(\hw)$'s signature $\preceq_2$ $\brf^{[\bru]}$'s signature, i.e. $\lpp(t\w) \preceq_2 \lpp(\bru)$ where $t=\lpp(\brf)/\lpp(h)$.
\end{enumerate}
}

The second condition may be a bit difficult to understand, but it is satisfied in all existing signature-based algorithms.

With Condition 1 and 2, we can prove Key Fact easily.

\begin{proof}[Proof of Key Fact]
We prove it by contradiction.

Assume $\lpp(f') \succ \lpp(g')$. Since a one-side-reduction is used in Key Fact, we have $\lpp(\u') = \lpp(\u) = \lpp(\v) = \lpp(\v')$. Let $\brf^{[\bru]} := f'^{[\u']} - c(g'^{[\v']})$ where $c=\lc(\u')/\lc(\v')$, then $\lpp(\brf) = \lpp(f')$ and $\lpp(\bru) \prec_2 \lpp(\u') = \lpp(\u)$. By Condition 2, there exists $h^{[\w]}\in G$ such that $\lpp(h)$ divides $\lpp(\brf)= \lpp(f')$ and $\lpp(t\w) \preceq_2 \lpp(\bru)\prec_2 \lpp(\u')$ where $t=\lpp(\brf)/\lpp(h)$.
This means $f'^{[\u']}$ is still reducible by $h^{[\w]}\in G$, which contracts with the definition of one-side-reduction.

The case $\lpp(f') \prec \lpp(g')$ can be proved similarly.
\end{proof}


\subsection{Observations}

Using Key Fact, we get the following important observations.

\vh\noindent
{\bf Observations 1:}
{\em If $\fu$ and $\gv$ have the same signature, then at most one of them is necessary to be reduced.}

\vh\noindent
{\bf Observations 2:}
{\em Particularly, if $\fu$ and $\gv$ have the same signature and either $f=0$ or $g=0$, then neither one is necessary to be reduced.}
\vh

We notice that all existing criteria are based on the above two observations. These observations also motivate the generalized criterion for signature-based algorithms.

\section{Generalized Criterion} \label{sec_criterion}

\subsection{Generalized criterion} \label{subsec_gencri}

Let $R:=\k[\x]$ and $\f:=(f_1, \cdots, f_m) \in R^m$. In the rest of paper, we consider the following ideal
$$I:=\langle f_1, \cdots, f_m\rangle=\{\u\cdot \f = p_1f_1+\cdots+p_mf_m \mid \u=(p_1,\cdots,p_m) \in R^m\}$$
with respect to some term order on $R$. The notation $\fu$ always means $f=\u\cdot \f$, and for convenience, we also call $\fu$ to be a polynomial in $I$ and write $\fu\in I$. Let $\e_i$ be the $i$-th unit vector of $R^m$, i.e. $(\e_i)_j=\delta_{ij}$ where $\delta_{ij}$ is the Kronecker delta. Then $f_1^{[\e_1]}, \cdots, f_m^{[\e_m]}$ are polynomials in $I$. Note that if there exists $\u'\not=\u$ such that $f=\u'\cdot \f$, then $\fu$ and $f^{[\u']}$ are treated as two different polynomials in $I$.

Fix {\em any} term order $\prec_1$ on $R$ and {\em any} term order $\prec_2$ on $R^m$. We must emphasize that the order $\prec_2$ may or may not be related to $\prec_1$ in theory, although $\prec_2$ is usually an extension of $\prec_1$ to $R^m$ in implementation. For sake of convenience, we use $\prec$ to represent $\prec_1$ and $\prec_2$, if no confusion occurs. We make the convention that if $f=0$ then $\lpp(f)=0$ and $0 \prec t$ for any non-zero power product $t$ in $R$; similarly for $\lpp(\u)$.

Given a finite set $B\subset I$, consider a {\bf partial order} ``$<$" defined on $B$, where ``$<$" has:
\begin{enumerate}

\item Non-Reflexivity: $\fu \nless \fu$ for all $\fu \in B$.

\item Antisymmetry: $\fu < \gv$ does not imply $\gv < \fu$, where $\fu, \gv\in B$.

\item Transitivity: $\fu < \gv $ and $\gv < \hw$ imply $\fu < \hw$, where $\fu, \gv, \hw \in B$.

\end{enumerate}

Now we give a generalized criterion for signature-based algorithms.

\begin{define}[generalized rewritable criterion]
Let $B$ be a subset of $I$, ``$<$" be a partial order on $B$, $\fu$ be a polynomial in $B$ and $t$ be a power product in $R$ where $f\not=0$. We say $t(\fu)$ is {\bf generalized rewritable} by $B$ ({\bf gen-rewritable} for short), if there exists $\gv \in B$ such that
\begin{enumerate}

\item $\lpp(\v)$ divides $\lpp(t\u)$, and

\item $\gv  < \fu$.

\end{enumerate}
\end{define}

If $t(\fu)$ is gen-rewritable by $\gv \in B$, then $\lpp(\v)$ divides $\lpp(t\u)$. Let $t':=\lpp(t\u)/\lpp(\v)$. Note that $t(\fu)$ and $t'(\gv)$ have the same signature, i.e. $\lpp(t\u) = \lpp(t'\v)$, so according to Observation 1, at most one of $t(\fu)$ and $t'(\gv)$ is necessary to be reduced during the computations. The partial order ``$<$" on $B$ will help to choose the polynomial that is not to be reduced, and in the above definition, the ``bigger" polynomial under the partial order is selected. So in practice, if $t(\fu)$ is gen-rewritable by $B$, then $t(\fu)$ will not be reduced.

Generally, the partial order on $B$ can be defined in many ways. For example, since the set $B$ is usually the intermediate set of generators and polynomials in $B$ are often added one by one,
then we can define ``$<$" as: $\gv < \fu, \mbox{ if } \gv \mbox{ is added to } B \mbox{ later than } \fu.$ There are two other partial orders: $\gv < \fu$ if $\lpp(g)< \lpp(f)$, or even $\gv < \fu$ if $f$ has more terms than $g$. All of these partial orders can be used in the above definition, but as we will see later, {\em not all partial orders lead to correct criterion}.


Observation 2 says if $\fu$ and $\gv$ have the same signature and either $f=0$ or $g=0$, then neither $\fu$ nor $\gv$ is necessary to be reduced. In fact, $0^{[\u]}$ means $\u$ is a syzygy of $\f=(f_1, \cdots, f_m)$, i.e. $\u\cdot \f =0$. For convenience, we call the polynomial $0^{[\u]}$ to be {\bf syzygy polynomial}. By using syzygy polynomials, the generalized criterion can be enhanced. That is, we can add syzygy polynomials to the set $B$ and assume syzygy polynomials are ``smaller" than other polynomials under the partial order, then more redundant computations can be rejected. This technique is used in the algorithm AGC in next subsection.

The following proposition shows many syzygy polynomials can be obtained directly.

\begin{prop}\label{prop_syzygy}
Let $\fu$ be a polynomial in $I$. Then $0^{[f\e_i-f_i\u]}$ is a syzygy polynomial where $1 \le i \le m$.
\end{prop}

\begin{proof} Since $f=\u \cdot (f_1, \cdots, f_m)$, then
$$(f\e_i - f_i\u) \cdot (f_1, \cdots, f_m) = f \e_i \cdot (f_1, \cdots, f_m) - f_i \u \cdot (f_1, \cdots, f_m) = f f_i - f_i f = 0.$$
\end{proof}
Since the syzygy polynomial $0^{[f\e_i-f_i\u]}$ in Proposition \ref{prop_syzygy} uses the principal syzygy of $f$ and $f_i$, we call syzygy polynomials in form of $0^{[f\e_i-f_i\u]}$ to be {\bf principle syzygy polynomials}.

In Section \ref{sec_specializations}, we will show how the generalized criterion specializes to F5 and GVW's criteria. Next, we describe how this generalized criterion is used in algorithm.

\subsection{How the generalized criterion is used?} \label{subsec_algorithm}

We first define the {\em critical pairs} of two polynomials. Suppose $\fu, \gv$ are two polynomials with $f$ and $g$ both nonzero. Let $t:=\lcm(\lpp(f), \lpp(g))$, $t_f:=t/\lpp(f)$ and $t_g:=t/\lpp(g)$. If $t_f(\fu)$'s signature $\succeq$ $t_g(\gv)$'s signature, i.e. $\lpp(t_f\u) \succeq \lpp(t_g\v)$, then the following 4-tuple vector $$(t_f, \fu, t_g, \gv)$$ is called the {\bf critical pair} of $\fu$ and $\gv$. The corresponding {\bf S-polynomial} is $t_f(\fu)-ct_g(\gv)$ where $c=\lc(f)/\lc(g)$. Please keep in mind that, for any critical pair $(t_f, \fu, t_g, \gv)$, we always have $t_f(\fu)$'s signature $\succeq$ $t_g(\gv)$'s signature, i.e. $\lpp(t_f\u) \succeq \lpp(t_g\v)$. Also note that $t_f$ (or $t_g$) here does not mean it only depends on $f$ (or $g$). For convenience, the critical pair of $\fu$ and $\gv$ is also denoted as $[\fu, \gv]$ or $[\gv, \fu]$ for short, and we say $[\fu, \gv]$ is a critical pair of $B$, if both $\fu$ and $\gv$ are in $B$.

Critical pairs can be classed in three kinds. Let $(t_f, \fu, t_g, \gv)$ be a critical pair and $t_f(\fu)-ct_g(\gv)$ be its S-polynomial where $c=\lc(f)/\lc(g)$.

\begin{enumerate}

\item If $t_f(\fu)-ct_g(\gv)$'s signature $\prec$ $t_f(\fu)$'s signature, i.e. $\lpp(t_f\u - ct_g\v) \not= \lpp(t_f\u)$, then we say $(t_f, \fu, t_g, \gv)$ is {\bf non-regular}.

\item If $t_f(\fu)-ct_g(\gv)$, $t_f(\fu)$ and $t_g(\gv)$ have the same signature, i.e. $\lpp(t_f\u - ct_g\v) = \lpp(t_f\u) = \lpp(t_g\v)$, then $(t_f, \fu, t_g, \gv)$ is called {\bf super regular}.

\item If $t_f(\fu)$'s signature $\succ$ $t_g(\gv)$'s signature, i.e. $\lpp(t_f\u) \succ \lpp(t_g\v)$, then we call $(t_f, \fu, t_g, \gv)$ to be {\bf regular}.

\end{enumerate}

We say a {\bf  critical pair  $(t_f, \fu, t_g, \gv)$ is gen-rewritable by a set $B$}, if {\em either} $t_f(\fu)$ {\em or} $t_g(\gv)$ is gen-rewritable by $B$.

Then the generalized criterion is used in the following way:

\vh{\em A critical pair $(t_f, \fu, t_g, \gv)$ of $B$ is {\bf\em rejected by the generalized criterion}, if
\begin{enumerate}
\item it is not regular, i.e. $\lpp(t_f\u) = \lpp(t_g\v)$, or

\item it is regular and generalized rewritable by $B$, i.e. $\lpp(t_f\u) \succ \lpp(t_g\v)$, and  either $t_f(\fu)$ or $t_g(\gv)$ is generalized rewritable by $B$.
\end{enumerate}
}

If a critical pair is rejected by the generalized criterion, then this critical pair will not be considered in algorithm. We can also show how the generalized criterion is used through a simple algorithm(Algorithm \ref{alg_agc}).

\begin{algorithm}[!ht]
\DontPrintSemicolon
\SetAlgoSkip{}
\SetAlgoNoLine
\SetKwInOut{Input}{Input}
\SetKwInOut{Output}{Output}
\SetKwFor{For}{for}{do}{end\ for}
\SetKwIF{If}{ElseIf}{Else}{if}{then}{else\ if}{else}{end\ if}

\Input{$f_1^{[\e_1]}, \cdots, f_m^{[\e_m]}$.}
\Output{A subset $G \subset \langle f_1, \cdots, f_m\rangle$.}

\Begin{
$G \lla \{f_i^{[\e_i]} \mid i=1, \cdots, m\} \cup \{0^{[f_j\e_i - f_i\e_j]} \mid 1 \leq i < j \leq m\}$ \spc\spc\spc\spc\spc $(\lozenge)$

$\mbox{\em CPairs} \lla \{[\fu, \gv] \mid \fu, \gv \in G\}$

\While{$\mbox{\sl CPairs} \not= \emptyset$}
{
  $[\fu, \gv] = (t_f, \fu, t_g, \gv) \lla $ {\bf any} critical pair in {\em CPairs} \spc\spc\spc $(\bigstar)$

  $\mbox{\em CPairs} \lla \mbox{\em CPairs} \setminus \{[\fu, \gv]\}$

  \If{$[\fu, \gv]$ is {\bf regular} and is {\bf not gen-rewritable} by $G$}
  {
    $\hw \lla$ reduce the S-polynomial of $[\fu, \gv]$ by $G$

    $\mbox{\em CPairs} \lla \mbox{\em CPairs} \cup \{ [\hw, \gv] \mid \gv \in G\}$

    $G \lla G \cup \{\hw\} \cup \{0^{[h\e_i - f_i\w]} \mid i =1,\cdots,m \}$ \spc\spc\spc\spc\spc\spc\spc $(\lozenge)$
  }
}
{\bf return} $G$
}
\caption{The algorithm with generalized criterion (AGC) \label{alg_agc}}
\end{algorithm}



For the above algorithm, please note that
 \begin{enumerate}
\item The gen-rewritable criterion uses a partial order defined on $G$. While new elements are added to $G$, the partial order on $G$ needs to be updated simultaneously. Fortunately, most partial orders can be updated automatically.

\item  For the line ended with ($\bigstar$), we emphasize that any critical pair can be selected, while some other algorithm, such as GVW, always selects the critical pair with minimal signature.

\item {\em Principle syzygy polynomials} are added to $G$ at lines marked with $(\lozenge)$.

\item The S-polynomial of $[\fu, \gv]$ is reduced by the one-side-reduction defined in Definition \ref{defn_reduce}. Note that for the reducing result $\hw$, we still have $h=\w\cdot (f_1, \cdots, f_m)$. Other similar one-side-reductions in \citep{Gao10b, Ars09, Fau02} can also be used here.

\item The S-polynomial of $(t_f, \fu, t_g, \gv)$ is considered only when $(t_f, \fu, t_g, \gv)$ is {\em regular}, so its S-polynomial $t_f(\fu)-ct_g(\gv)$ and $t_f(\fu)$ have the same signature where $c=\lc(f)/\lc(g)$. Besides, the one-side-reduction does not affect the signatures, i.e. $t_f(\fu)$ and $\hw$ also have the same signature. Therefore, for sake of efficiency, it suffices to record $f$ and $\lpp(\u)$ for each $\fu\in G$ in the practical implementation, which is just the same as that F5 does.
\end{enumerate}

The algorithm AGC aims to compute a \gr basis for $\langle f_1, \cdots, f_m\rangle$. However, the generalized criterion may reject useful critical pairs sometimes, which makes the output of the algorithm AGC is not a \gr basis. In next subsection, we will show when the generalized criterion is correct, or equivalently, when the algorithm AGC outputs a \gr basis.

\subsection{When the generalized criterion is correct?}

In fact, the algorithm AGC can construct a even ``stronger" version of \gr basis. Let $$G := \{g_1^{[\v_1]}, \cdots, g_s^{[\v_s]}\}$$ be a finite subset of $I$. We call $G$ a {\bf labeled Gr\"obner basis}\footnote{Labeled \gr basis is exactly the same as the S-\gr basis in \citep{SunWang11}, and it is also a simpler version of {\em strong \gr basis} defined in \citep{Gao10b}, so the GVW algorithm computes a labeled Gr\"obner basis. We proved in another paper that F5 also computes a labeled \gr basis.} for $I$, if for any $\fu \in I$ with $f\not=0$, there exists $\gv\in G$ such that
\begin{enumerate}

\item $\lpp(g)$ divides $\lpp(f)$, and

\item $t(\gv)$'s signature $\preceq$ $\fu$'s signature, i.e. $\lpp(t\v) \preceq \lpp(\u)$, where $t=\lpp(f)/\lpp(g)$.

\end{enumerate}

\begin{prop} \label{prop_gb}
If $G$ is a labeled \gr basis for $I$, then the set $\{g \mid \gv \in G\}$ is a \gr basis of the ideal $I=\langle f_1, \cdots, f_m\rangle$.
\end{prop}

\begin{proof}
For any $f\in \langle f_1, \cdots, f_m\rangle$, there exist $p_1, \cdots, p_m \in R$ such that $f=p_1f_1 + \cdots + p_m f_m$. Let $\u:=(p_1, \cdots, p_m)$. Then $\fu \in I$ and hence there exists $\gv\in G$ such that $\lpp(g)$ divides $\lpp(f)$ by the definition of labeled \gr basis.
\end{proof}


The algorithm AGC outputs a labeled \gr basis for $I$, if the partial order in the generalized criterion is {\em admissible}. In the algorithm AGC, we say a partial order ``$<$" is {\bf admissible}, if for any critical pair $(t_f, \fu, t_g, \gv)$ of $G$, whenever we need to reduce the S-polynomial of $(t_f, \fu, t_g, \gv)$ to $\hw$ by $G$, we always have $\hw < \fu$ after updating ``$<$" for $G \cup \{\hw\}$. In next section, we will show that the partial orders implied by F5 and GVW's criteria are both admissible.

Note that only the critical pair {\em that is regular and not gen-rewritable} is really reduced in the algorithm AGC, so when checking whether a partial order is admissible, we do not care about the critical pairs that are rejected by the generalized criterion. Besides, we emphasize that in the above definition of admissible, the relation $\hw < \fu$ is essential, and $\hw$ may not be related to other elements in $G$.



The following theorem shows when the generalized criterion is correct in the algorithm AGC. The proof of theorem will be presented in Section \ref{sec_proof}.

\begin{theorem}  \label{thm_main}
Let $I := \langle f_1, \cdots, f_m \rangle$ be an ideal in $R$. Then a labeled \gr basis for $I$ can be constructed by the algorithm AGC, if the algorithm AGC terminates in finite steps and the partial order in the generalized criterion is admissible.
\end{theorem}

\section{Specializations} \label{sec_specializations}

In this section, we focus on specializing the generalized criterion to F5 and GVW's criteria by using appropriate admissible partial orders. By saying ``specialize", we mean the critical pairs rejected by F5 or GVW's criteria can also be rejected by the generalized criterion.


\subsection{F5's criteria}

First, we list the F5's criteria with current notations. In F5, the order $\prec_2$ on $R^m$ is obtained by extending $\prec_1$ to $R^m$ in a {\em position over term} fashion, i.e.
$$x^\alpha\e_i \prec_2 x^\beta\e_j    \mbox{  iff  } \left\{\begin{array}{l} i > j, \\ \mbox{ or } \\ i = j \mbox{ and } x^\alpha \prec_1 x^\beta. \end{array}\right.$$
This term order makes F5 work incrementally.

\begin{define}[syzygy criterion]
Let $B$ be a subset of $I$, $\fu$ be a polynomial in $B$ and $t$ be a power product in $R$ where $f\not=0$ and $\lpp(\u) = x^\alpha \e_i$. We say $t(\fu)$ is {\bf F5-divisible} by $B$, if there exists $\gv \in B$ with $\lpp(\v)=x^\beta \e_j$, such that
\begin{enumerate}

\item $\lpp(g)$ divides $tx^\alpha$, and

\item $\e_i \succ \e_j$.

\end{enumerate}
\end{define}

\begin{define}[rewritten criterion]
Let $B$ be a subset of $I$, $\fu$ be a polynomial in $B$ and $t$ be a power product in $R$. We say $t(\fu)$ is {\bf F5-rewritable} by $B$, if there exists $\gv \in B$, such that
\begin{enumerate}

\item $\lpp(\v)$ divides $\lpp(t\u)$, and

\item $\gv$ is added to $B$ later than $\fu$.

\end{enumerate}
\end{define}

In F5, a critical pair $(t_f, \fu, t_g, \gv)$ of $B$ is rejected by the syzygy criterion or rewritten criterion if either $t_f(\fu)$ or $t_g(\gv)$ is F5-divisible or F5-rewritable by $B$.

Next, we show how the generalized criterion specializes to both syzygy criterion and rewritten criterion at the same time. For this purpose, the following partial order on $G$, which can be updated automatically when a new element is added to $G$, is used: For any $\fu, \gv\in G$, we say $\gv < \fu$ if
\begin{enumerate}

\item $f\not=0$ and $\gv = 0^{[\v]}$ is a {\em principle syzygy polynomial},

\item otherwise, $\gv$ is added to $G$ later than $\fu$.

\end{enumerate}
The above partial order ``$<$" is {\em admissible} in the algorithm AGC. Because for any critical pair $(t_f, \fu, t_g, \gv)$ of $G$, when we need to reduce its S-polynomial to $\hw$ by $G$, the polynomial $\hw$ is always added to $G$ later than $\fu$ no matter $h$ is $0$ or not, since $\fu$ is already in $G$.

At last, we show how the generalized criterion specializes to the rewritten criterion and
syzygy criterion. For the rewritten criterion, the specialization is obvious by the definition
of ``$<$". For the syzygy criterion, if $t(\fu)$, where $\fu\in G$ with $\lpp(\u) = x^\alpha \e_i$ and $f \not= 0$, is F5-divisible by some $\gv \in G$ with $\lpp(\v)=x^\beta \e_j$, we have $\lpp(g)$
divides $tx^\alpha$ and $\e_i \succ \e_j$. Since $\gv \in G$, ccording to the algorithm AGC, the principle syzygy polynomial $0^{[g\e_i - f_i\v]}$ has been added to $G$, and $\lpp(g\e_i - f_i\v) = \lpp(g)\e_i$ divides $t x^\alpha \e_i$. So $t(\fu)$ is gen-rewritable by $0^{[g\e_i - f_i\v]} \in G$.
Therefore, the critical pairs rejected by F5's criteria can also be rejected by the generalized criterion.

With a similar discussion, the generalized criterion can also specialize to the criteria in \citep{Ars09}, since the extended F5 algorithm in that paper only differs from the original F5 in the order $\prec_2$ on $R^m$.

\subsection{GVW's Criteria}

First, we rewrite the GVW's criteria with current notations.

\begin{define}[First Criterion]
Let $B$ be a subset of $I$, $\fu$ be a polynomial in $B$ and $t$ be a power product in $R$ where $f\not=0$. We say $t(\fu)$ is {\bf GVW-divisible} by $B$, if there exists $\gv \in B$ such that
\begin{enumerate}

\item $\lpp(\v)$ divides $\lpp(t\u)$, and

\item $g=0$.

\end{enumerate}
\end{define}

\begin{define}[Second Criterion]
Let $B$ be a subset of $I$, $\fu$ be a polynomial in $B$ and $t$ be a power product in $R$. We say $t(\fu)$ is {\bf eventually super top-reducible} by $B$, if $t(\fu)$ is reducible and can be reduced to $\hw$ by $B$, and there exists $\gv \in B$ such that
\begin{enumerate}

\item $\lpp(\v)$ divides $\lpp(\w)$,

\item $\lpp(g)$ divides $\lpp(h)$, and

\item $\frac{\lpp(\w)}{\lpp(\v)} = \frac{\lpp(h)}{\lpp(g)}$ and $\frac{\lc(\w)}{\lc(\v)} = \frac{\lc(h)}{\lc(g)}$.

\end{enumerate}
\end{define}

In GVW, a critical pair $(t_f, \fu, t_g, \gv)$ of $B$ is rejected, if $t_f(\fu)$ is GVW-divisible or eventually super top-reducible by $B$. The GVW algorithm also has a third criterion.

\vh\noindent{\bf Third Criterion}
{\em If there are two critical pairs $(t_f, \fu, t_g, \gv)$ and $({t}_{\brf}, \brfu, {t}_{\brg}, \brgv)$ of $B$ such that $t_f(\fu)$ and ${t}_{\brf}(\brfu)$ have the same signature, i.e. $\lpp(t_f\u) = \lpp({t}_{\brf} \bar{\u})$, then at least one of the two critical pairs is redundant.}
\vh

Next, in order to specialize the generalized criterion to the above three criteria at the same time, the following partial order on $G$, which can also be updated automatically when a new element is added to $G$, is used: for any  $\fu, \gv \in G$, we say $\gv < \fu$ if one of the following two conditions holds:
\begin{enumerate}

\item  $\lpp(t'g) < \lpp(tf)$, where $t'= \frac{\lcm(\lpp(\u), \lpp(\v))}{\lpp(\v)}$ and $t= \frac{\lcm(\lpp(\u), \lpp(\v))}{\lpp(\u)}$ such that $t(\fu)$ and $t'(\gv)$ have the same signature, i.e. $\lpp(t\u) = \lpp(t'\v)$.

\item $\lpp(t'g) = \lpp(tf)$ and $\gv$ is added to $G$ later than $\fu$.

\end{enumerate}
The above partial order ``$<$" is {\em admissible} in the algorithm AGC. Because for any critical pair $(t_f, \fu, t_g, \gv)$ of $G$, when we need to reduce its S-polynomial to $\hw$ by $G$, we always have $\lpp(t_f\u) = \lpp(\w)$ and $\lpp(t_ff) > \lpp(h)$.

At last, let us see the three criteria of GVW.

For the first criterion, if $t(\fu)$ is GVW-divisible by some $\gv \in G$, then $t(\fu)$ is also gen-rewritable by $\gv \in G$ by definition.

For the second criterion, if $t(\fu)$, where $\fu \in G$, is eventually super top-reducible by $G$, then $t(\fu)$ can be reduced to $\hw$ and there exists $\gv \in G$ such that $\lpp(\v)$ divides $\lpp(\w)$, $\lpp(g)$ divides $\lpp(h)$, $\frac{\lpp(\w)}{\lpp(\v)} = \frac{\lpp(h)}{\lpp(g)}$ and $\frac{\lc(\w)}{\lc(\v)} = \frac{\lc(h)}{\lc(g)}$. Then we have $\lpp(t'g) = \lpp(h) < \lpp(tf)$ and $\lpp(t'\v) = \lpp(\w) = \lpp(t\u)$, which means $\gv < \fu$. So $t(\fu)$ is gen-rewritable by $\gv \in G$.

For the third criterion, we have $\lpp(t_f\u) = \lpp({t}_{\brf}\bar{\u})$. Note that the above partial order is in fact a total order. First, if $\fu < \brfu$, then ${t}_{\brf}(\brfu)$ is gen-rewritable by $\fu$ and hence $({t}_{\brf}, \brfu, {t}_{\brg}, \brgv)$ is rejected; the reverse is also true. Second, if $\fu = \brfu$, then one of the two critical pairs should be selected earlier from the set $CPairs$, assuming $(t_f, \fu, t_g, \gv)$ is selected first. On one hand, if $(t_f, \fu, t_g, \gv)$ is regular and not gen-rewritable, then its S-polynomial is reduced to $\hw$ and $\hw$ is added to $G$ by the algorithm AGC. Since ``$<$" is admissible, we have $\hw < \fu$. Thus, when the critical pair $({t}_{\brf}, \brfu, {t}_{\brg}, \brgv)$ is selected afterwards, it will be rejected, since ${t}_{\brf}(\brfu)$ is gen-rewritable by $\hw$. On the other hand, if $(t_f, \fu, t_g, \gv)$ is not regular, or it is regular and gen-rewritable, then $(t_f, \fu, t_g, \gv)$ is rejected at once. Anyway, at least one of the two critical pairs is rejected in the algorithm.


\section{Proofs for the Correctness of the Generalized Criterion} \label{sec_proof}

To prove Theorem \ref{thm_main}, we need the following definition and lemmas. 

In this section, we always assume that $I$ is the ideal generated by $\{f_1, \cdots, f_m\}$. Let $\fu\in I$, we say $\fu$ has a {\bf standard representation} w.r.t. a set $B\subset I$, if there exist $p_1, \cdots, p_s \in R$ and $g_1^{[\v_1]}, \cdots, g_s^{[\v_s]}\in B$ such that $$f = p_1 g_1 + \cdots + p_s g_s,$$ where $\lpp(f)\succeq \lpp(p_ig_i)$ and $\fu$'s signature $\succeq$ $p_i(g_i^{[\v_i]})$'s signature, i.e. $\lpp(\u) \succeq \lpp(p_i\v_i)$ for $i=1,\cdots, s$. Clearly, if $\fu$ has a  standard representation w.r.t. $B$, then there exists $\gv\in B$ such that $\lpp(g)$ divides $\lpp(f)$ and $\lpp(\u) \succeq \lpp(t\v)$ where $t=\lpp(f)/\lpp(g)$. We call this property to be the {\bf basic property} of standard representations.

\begin{lemma} \label{lem_stdrepresentation}
Let $G$ be a finite subset of $I$ and $\{f_1^{[\e_1]},$ $\cdots, f_m^{[\e_m]}\}\subset G$. For a polynomial $\fu \in I$, $\fu$ has a standard representation w.r.t. $G$, if for any critical pair $[\gv, \hw]=(t_g, \gv, t_h, \hw)$ of $G$ with $\fu$'s signature $\succeq$ $t_g(\gv)$'s signature, i.e. $\lpp(\u) \succeq \lpp(t_g \v)$, the S-polynomial of $[\gv, \hw]$  always has a standard representation w.r.t. $G$.
\end{lemma}

\begin{proof}
For $\fu\in I$, we have $\u\cdot \f = f$ where $\f = (f_1, \cdots, f_m)\in R^m$. Assume $\u=p_1\e_1+\cdots+p_m\e_m$  where $p_i\in R$. Clearly, $f = p_1 f_1 + \cdots + p_m f_m.$ Note that $\lpp(\u)\succeq \lpp(p_i\e_i)$ for $i=1,\cdots,m$. If $\lpp(f)\succeq \lpp(p_if_i)$, then we have already got a standard representation for $\fu$ w.r.t. $G$. Otherwise, we will prove it through classical method. Let $T := \max\{\lpp(p_if_i)\mid i=1,\cdots,m\}$, then $T\succ \lpp(f)$ holds by assumption.
Consider the equation
$$f= \sum_{\lpp(p_if_i)=T} \lc(p_i)\lpp(p_i) f_i +\sum_{\lpp(p_jf_j)\prec T}p_j f_j+ \sum_{\lpp(p_if_i)=T} (p_i - \lc(p_i)\lpp(p_i)) f_i.\eqno(1)$$
The leading power products in the first sum should be canceled, since we have $T\succ \lpp(f)$. So the first sum can be rewritten as a sum of S-polynomials, that is $$\sum_{\lpp(p_if_i)=T}\lc(p_i)\lpp(p_i) f_i= \sum \bar{c}t(t_g g-ct_h h),$$
 where $\gv, \hw\in G$, $\bar{c}\in \k$, $t_g(\gv)-ct_h(\hw)$ is the S-polynomial of $(t_g, \gv, t_h, \hw)$, $\lpp(t\ t_g g)=\lpp(t\ t_h h)=T$ and $\lpp(\u)\succeq \lpp(t\ t_g \v)\succeq \lpp(t\ t_h\w)$ such that we have $\lpp(t ( t_g g - c  t_h h)) \prec T $. By the hypothesis of the lemma, the S-polynomial $t_g(\gv)-ct_h(\hw) = (t_g g-c t_h h)^{[t_g\v-c t_h \w]}$ has a standard representation w.r.t. $G$, that is, there exist $g_i^{[\v_i]} \in G$, such that $ t_g g-c t_h h = \sum q_i g_i$, where $\lpp(t_g g-c t_h h)\succeq \lpp(q_i g_i)$ and $\lpp(\u)\succeq \lpp(t\ t_g\v)\succeq \lpp(t\ q_i \v_i)$. Substituting these standard representations back to the original expression of $f$ in $(1)$, we get a new representation for $f$. Let $T^{(1)}$ be the maximal leading power product of the polynomials appearing in the right side of the new representation. Then we have $T\succ T^{(1)}$.  Repeat the above process until  $T^{(s)}$ is the same as $\lpp(f)$ for some $s$ after finite steps. Finally, we always get a standard representation for $\fu$.
\end{proof}

\begin{lemma} \label{lem_correctness}
Let $G$ be a finite subset of $I$ and $\{f_1^{[\e_1]},$ $\cdots, f_m^{[\e_m]}\}\subset G$. Then $G$ is a labeled Gr\"obner basis for $I$, if for any critical pair $[\fu, \gv]$ of $G$, the S-polynomial of $[\fu, \gv]$   always has a standard representation w.r.t. $G$.
\end{lemma}

\begin{proof}
Using Lemma \ref{lem_stdrepresentation}, for any $\fu\in I$, $\fu$ has a standard representation w.r.t. $G$. By the basic property of standard representations, $G$ is a labeled Gr\"obner basis for $I$.
\end{proof}

Before giving a full proof of Theorem \ref{thm_main}, we introduce the following definitions.

Suppose $(t_f, \fu, t_g, \gv)$ and $({t}_{\brf}, \brfu, {t}_{\brg}, \brgv)$ are two critical pairs, we say $({t}_{\brf}, \brfu, {t}_{\brg}, \brgv)$ is {\bf smaller} than $(t_f, \fu$, $t_g, \gv)$  if one of the following conditions holds:
\begin{enumerate}

\item[(a).] $\lpp({t}_{\brf} \bru) \prec \lpp(t_f\u)$.

\item[(b).] $\lpp({t}_{\brf} \bru) = \lpp(t_f\u)$ and $\brfu < \fu$.

\item[(c).] $\lpp({t}_{\brf} \bru) = \lpp(t_f\u)$, $\brfu = \fu$ and $\lpp({t}_{\brg} \brv) \prec \lpp(t_g\v)$.

\item[(d).] $\lpp({t}_{\brf} \bru) = \lpp(t_f\u)$, $\brfu = \fu$, $\lpp({t}_{\brg} \brv) = \lpp(t_g\v)$ and $\brgv < \gv$.

\end{enumerate}

Let $D$ be a set of critical pairs. A critical pair in $D$ is said to be  {\bf minimal}  if there is no critical pair in $D$  smaller  than this critical pair. Remark that the order ``smaller" defined on the critical pairs is {\em a partial order}, i.e. some critical pairs may not be comparable. Thus, the minimal critical pair in $D$ may not be unique, but we can always find one if $D$ is finite.

Now, we give the proof of Theorem \ref{thm_main}.

\begin{proof}[Proof of Theorem \ref{thm_main}] Let $G_{end}$ denote the set returned by the algorithm AGC. According to the hypotheses, $G_{end}$ is finite, and we also have $\{f_1^{[\e_1]}, \cdots, f_m^{[\e_m]}\} \subset G_{end}$ by the algorithm AGC. To show $G_{end}$ is a labeled \gr basis for $I$, we will take the following strategy.\vh \\
{\bf Step 1:}  Let $Todo$ be the set of  {\em all} the critical pairs of $G_{end}$, and $Done$ be an empty set.\vh \\
{\bf Step 2:} Select a minimal critical pair $[\fu, \gv]=(t_f, \fu, t_g, \gv)$ in $Todo$. \vh \\
{\bf Step 3:} For such $[\fu, \gv]$,  we will prove the following two facts.
\begin{enumerate}
\item[(F1).] The  S-polynomial of  $[\fu, \gv]$    has a standard representation w.r.t. $G_{end}$.

\item[(F2).] If $(t_f, \fu, t_g, \gv)$ is {\em super regular} or {\em regular}, then $t_f(\fu)$ is gen-rewritable by $G_{end}$.
\end{enumerate}
{\bf Step 4:} Move $[\fu, \gv]$ from $Todo$ to $Done$, i.e. $Todo \lla Todo \setminus \{[\fu, \gv]\}$ and $Done \lla Done\ \cup \{ [\fu, \gv]\}$. \vh \\
We can repeat {\bf Step 2, 3, 4} until $Todo$ is empty. Please note that for every critical pair in $Done$, it always has property (F1); particularly, if this critical pair is super regular or regular, then it has properties  (F1) and (F2). When $Todo$ is empty, all the critical pairs of $G_{end}$ will lie in $Done$, and hence, all the corresponding S-polynomials  have standard representations w.r.t. $G_{end}$. Then $G_{end}$ is a labeled Gr\"obner basis by Lemma \ref{lem_correctness}.

{\bf Step 1, 2, 4} are trivial, so we next focus on showing the facts in {\bf Step 3}.

Take a minimal critical pair $[\fu, \gv] = (t_f, \fu, t_g, \gv)$ in $Todo$. And this critical pair must appear in the algorithm AGC. Suppose such pair is selected from the set $CPairs$ in  some  loop of the algorithm AGC and $G_k$ denotes the corresponding set $G$ at the beginning of that loop. For such $[\fu, \gv]$, it must be in one of the following cases:
\begin{enumerate}

\item[C1:] $[\fu, \gv]$ is {\em non-regular}.

\item[C2:] $[\fu, \gv]$ is {\em super regular}.

\item[C3:] $[\fu, \gv]$ is {\em regular} and is {\em not} gen-rewritable by $G_k$.

\item[C4:] $[\fu, \gv]$ is {\em regular} and $t_f(\fu)$ is gen-rewritable by $G_k$.

\item[C5:] $[\fu, \gv]$ is {\em regular} and $t_g(\gv)$ is gen-rewritable by $G_k$.

\end{enumerate}
Thus, to show the facts in {\bf Step 3}, we have two things to do: First, show (F1) holds in case {\bf C1}; Second, show both (F1) and (F2) hold in cases {\bf C2, C3, C4} and {\bf C5}.

We  make the following claims under the condition that $[\fu, \gv] = (t_f, \fu, t_g, \gv)$ is  minimal in $Todo$. The proofs of these claims will be presented after the current proof.

{\bf Claim 1}: For any $\brfu\in I$, if $\brfu$'s signature $\prec$ $t_f(\fu)$'s signature, i.e. $\lpp(\bru) \prec \lpp(t_f\u)$, then $\brfu$ has a standard representation w.r.t. $G_{end}$.

{\bf Claim 2}: If $[\fu, \gv]$ is super regular or regular and $t_f(\fu)$ is gen-rewritable by $G_{end}$, then the S-polynomial of $[\fu, \gv]$ has a standard representation w.r.t. $G_{end}$.

{\bf Claim 3:} If $[\fu, \gv]$ is regular and $t_g(\gv)$ is gen-rewritable by  $G_{end}$, then $t_f(\fu)$ is also gen-rewritable by $G_{end}$.

Note that {\bf Claim 2} plays an important role in the whole proof. Since {\bf Claim 2} shows that (F2) implies (F1) in the cases {\bf C2, C3, C4} and {\bf C5}, it suffices to show $t_f(\fu)$ is gen-rewritable by $G_{end}$ in these cases.

Next, we proceed with each case respectively.

{\bf C1:} $[\fu, \gv]$ is {\em non-regular}. Consider the S-polynomial $t_f(\fu) - c t_g (\gv) = (t_f f-c t_g g)^{[t_f\u-ct_g\v]}$ where $c=\lc(f)/\lc(g)$. Note that $\lpp(t_f\u-c t_g\v) \prec \lpp(t_f\u)$ by the definition of non-regular, so {\bf Claim 1} shows $(t_f f-c t_g g)^{[t_f\u-ct_g\v]}$ has a standard representation w.r.t. $G_{end}$, which proves (F1).

{\bf C2:} $[\fu, \gv]$ is {\em super regular}, i.e. $\lpp(t_f\u - ct_g\v) = \lpp(t_f\u) = \lpp(t_g\v)$ where $c=\lc(f)/\lc(g)$. Let $\bar{c} := \lc(\u)/\lc(\v)$. Note that $\bar{c}\not=c$, since $\lpp(t_f\u-c t_g \v)=\lpp(t_f\u)$. Then we have $\lpp(t_f f-\bar{c} t_g g)=\lpp(t_f f)$ and $\lpp(t_f \u-\bar{c} t_g \v)\prec \lpp(t_f\u)$. So {\bf Claim 1} shows $t_f(\fu)- \bar{c} t_g (\gv) = (t_f f-\bar{c} t_g g)^{[t_f\u-\bar{c} t_g \v]}$ has a standard representation w.r.t. $G_{end}$, and hence, there exists $\hw\in G_{end}$ such that $\lpp(h)$ divides $\lpp(t_f f - \bar{c} t_g g)=\lpp(t_f f)$ and $\lpp(t_f \u)\succ \lpp(t_f \u-\bar{c} t_g \v) \succeq \lpp(t_h \w)$ where $t_h= \lpp(t_f f)/\lpp(h)$. Next, consider the critical pair $[\fu, \hw]$. Since $\lpp(t_f f) = \lpp(t_h h)$, the critical pair $[\fu, \hw]$ has two possible forms. \smallskip\\
{Form 1:} $[\fu, \hw] = (t_f, \fu, t_h, \hw)$. Since $\lpp(t_f \u) = \lpp(t_g\v) \succ \lpp(t_h \w)$, the critical pair $[\fu, \hw]$ is regular and is smaller than $(t_f, \fu, t_g, \gv)$ in fashion (c), which means $[\fu, \hw]$ lies in $Done$ and $t_f(\fu)$ is gen-rewritable by $G_{end}$.
\smallskip\\
{Form 2:} $[\fu, \hw] = (\bar{t}_f, \fu, \bar{t}_h, \hw)$ where $\bar{t}_f$ divides $t_f$ and $\bar{t}_f\not= t_f$. Since $\lpp(t_f \u)\succ \lpp(t_h \w)$, the critical pair $(\bar{t}_f, \fu, \bar{t}_h, \hw)$ is also regular and is smaller than $(t_f, \fu, t_g, \gv)$ in fashion (a), which means $(\bar{t}_f, \fu, \bar{t}_h, \hw)$ lies in $Done$ and $\bar{t}_f(\fu)$ is gen-rewritable by $G_{end}$. Then ${t}_f(\fu)$ is also gen-rewritable by $G_{end}$, since $\bar{t}_f$ divides $t_f$.
\smallskip

{\bf C3:} $[\fu, \gv]$ is {\em regular} and {\em not} gen-rewritable by $G_k$. According to the algorithm AGC, the S-polynomial $t_f(\fu)-c t_g(\gv)$ is reduced to $\hw$ by $G_k$ where $c=\lc(f)/\lc(g)$, and $\hw$ will be added to the set $G_k$ at the end of this loop. Note that $G_k\subset G_{end}$ and $\hw\in G_{end}$. Since ``$<$" is an admissible partial order, we have $\hw<\fu$ by definition. Combined with the fact $\lpp(\w) = \lpp(t_f\u)$, so ${t}_f(\fu)$ is gen-rewritable by $\hw\in G_{end}$.

{\bf C4}: $[\fu, \gv]$ is {\em regular} and $t_f(\fu)$ is gen-rewritable by $G_k$. Then $t_f(\fu)$ is also gen-rewritable by $G_{end}$, since $G_k\subset G_{end}$.

{\bf C5}: $[\fu, \gv]$ is {\em regular} and $t_g(\gv)$ is gen-rewritable by $G_k$.  $t_g(\gv)$ is also gen-rewritable by $G_{end}$, since $G_k \subset G_{end}$. Then {\bf Claim 3} shows $t_f(\fu)$ is gen-rewritable by $G_{end}$ as well.

Theorem \ref{thm_main} is proved.
\end{proof}

We give the proofs for the three claims below.

\begin{proof}[Proof of {\bf Claim 1}]
According to the hypothesis, we have $\brfu\in I$ and $\lpp(\bru)\prec \lpp(t_f\u)$. So for any critical pair $(t_{f'}, {f'}^{[\u']}, t_{g'}, {g'}^{[\v']})$ of $G_{end}$ with $\lpp(\bru) \succeq \lpp(t_{f'} \u')$, the critical pair $(t_{f'}, {f'}^{[\u']}$, $t_{g'}, {g'}^{[\v']})$ is smaller than $(t_f, \fu, t_g, \gv)$ in fashion (a) and hence lies in $Done$, which means the S-polynomial of $(t_{f'}, {f'}^{[\u']}, t_{g'}, {g'}^{[\v']})$ has a standard representation w.r.t. $G_{end}$. So Lemma \ref{lem_stdrepresentation} shows that $\brfu$ has a standard representation w.r.t. $G_{end}$.
\end{proof}

\begin{proof}[Proof of {\bf Claim 2}]
We have that $[\fu, \gv] = (t_f, \fu, t_g, \gv)$ is minimal in $Todo$ and $t_f(\fu)$ is gen-rewritable by $G_{end}$. Let $c := \lc(f)/\lc(g)$. Then $\brfu = t_f (\fu) - c t_g (\gv) = (t_f f - c t_g g)^{[t_f\u - c t_g \v]}$ is the S-polynomial of $[\fu, \gv]$. Since $[\fu, \gv]$ is super regular or regular, we have $\lpp(\bru) = \lpp(t_f \u)$. Next we will show that $\brfu$ has a standard representation w.r.t. $G_{end}$. The proof is organized as follows. \vh\\
{\bf First:} We show that there exists $f_0^{[\u_0]}\in G_{end}$ such that $f_0^{[\u_0]} < \fu$, $t_f(\fu)$ is gen-rewritable by $f_0^{[\u_0]}$ and $t_0(f_0^{[\u_0]})$ is {\em not} gen-rewritable by $G_{end}$ where $t_0 = \lpp(t_f\u)/\lpp(\u_0)$.\vh \\
{\bf Second:} For such $f_0^{[\u_0]}$, we show that $\lpp(\brf) \succeq \lpp(t_0f_0)$ where $t_0 = \lpp(t_f\u)/\lpp(\u_0)$.\vh\\
{\bf Third:} We prove that $\brfu$ has a standard representation w.r.t. $G_{end}$.\vh

Proof of the {\bf First} fact. By hypothesis, suppose $t_f(\fu)$ is gen-rewritable by some $f_1^{[\u_1]}\in G_{end}$, i.e. $\lpp(\u_1)$ divides $\lpp(t_f\u)$ and $f_1^{[\u_1]} < \fu$. Let $t_1 := \lpp(t_f \u) /\lpp(\u_1)$. If $t_1(f_1^{[\u_1]})$ is not gen-rewritable by $G_{end}$, then $f_1^{[\u_1]}$ is the polynomial we are looking for. Otherwise, there exists $f_2^{[\u_2]}\in G_{end}$ such that $t_1(f_1^{[\u_1]})$ is gen-rewritable by $f_2^{[\u_2]}$. Note that $t_f(\fu)$ is also gen-rewritable by $f_2^{[\u_2]}$ and we have $\fu > f_1^{[\u_1]} > f_2^{[\u_2]}$. Let $t_2 := \lpp(t_f \u) /\lpp(\u_2)$. We next discuss whether $t_2(f_2^{[\u_2]})$ is gen-rewritable by $G_{end}$. In the better case, $f_2^{[\u_2]}$ is the desired polynomial if $t_2(f_2^{[\u_2]})$ is not gen-rewritable by $G_{end}$; while in the worse case, $t_2(f_2^{[\u_2]})$ is gen-rewritable by some $f_3^{[\u_3]} \in G_{end}$. We can repeat the above discussions for the worse case. Finally, we will get a chain $\fu > f_1^{[\u_1]} > f_2^{[\u_2]} > \cdots$. This chain must terminate, since $G_{end}$ is finite and ``$>$" is a partial order defined on $G_{end}$. Suppose $f_s^{[\u_s]}$ is the last one in the above chain. Then $t_f(\fu)$ is gen-rewritable by $f_s^{[\u_s]}$ and $t_s(f_s^{[\u_s]})$ is not gen-rewritable by $G_{end}$ where $t_s = \lpp(t_f\u)/\lpp(\u_s)$.

Proof of the {\bf Second} fact. From the {\bf First} fact, we have that $t_0(f_0^{[\u_0]})$ is {\em not} gen-rewritable by $G_{end}$ where $t_0 = \lpp(t_f\u)/\lpp(\u_0)$. Next, we prove the {\bf Second} fact by contradiction. Assume $\lpp(\brf) \prec \lpp(t_0f_0)$. Let $c_0 := \lc(\bru)/\lc(\u_0)$. Then for the polynomial $\brfu - c_0 t_0 (f_0^{[\u_0]}) = (\brf - c_0 t_0 f_0)^{[\bru - c_0 t_0 \u_0]}$, we have $\lpp(\brf - c_0 t_0 f_0) = \lpp(t_0 f_0)$ and $\lpp(\bru - c_0 t_0 \u_0) \prec \lpp(\bru) = \lpp(t_0\u_0) = \lpp(t_f\u)$. So $(\brf - c_0 t_0 f_0)^{[\bru - c_0 t_0 \u_0]}$ has a standard representation w.r.t. $G_{end}$ by {\bf Claim 1}, and hence, there exists $\hw\in G_{end}$ such that $\lpp(h)$ divides $\lpp(\brf - c_0 t_0 f_0) = \lpp(t_0 f_0)$ and $\lpp(t_0\u_0) \succ \lpp(\bru - c_0 t_0 \u_0) \succeq \lpp(t_h\w)$ where $t_h=\lpp(t_0f_0)/\lpp(h)$. Next consider the critical pair $[f_0^{[\u_0]}, \hw]$. Similarly, since $\lpp(t_0f_0) = \lpp(t_h h)$, the critical pair $[f_0^{[\u_0]}, \hw]$ has two possible forms.
\smallskip\\
{Form 1:} $[f_0^{[\u_0]}, \hw] = (t_0, f_0^{[\u_0]}, t_h, \hw)$. Since $\lpp(t_0\u_0)\succ \lpp(t_h \w)$, the critical pair $[f_0^{[\u_0]}, \hw]$ is regular and is smaller than $(t_f, \fu, t_g, \gv)$ in fashion (b), which means $[f_0^{[\u_0]}, \hw]$ lies in $Done$ and $t_0(f_0^{[\u_0]})$ is gen-rewritable by $G_{end}$, which contradicts with the property that $t_0(f_0^{[\u_0]})$ is {\em not} gen-rewritable by $G_{end}$.
\smallskip\\
{Form 2:} $[f_0^{[\u_0]}, \hw] = (\bar{t}_0, f_0^{[\u_0]}, \bar{t}_h, \hw)$ where $\bar{t}_0$ divides $t_0$ and $\bar{t}_0 \not= t_0$. Since $\lpp(t_0\u_0) \succ \lpp(t_h \w)$, the critical pair $(\bar{t}_0, f_0^{[\u_0]}, \bar{t}_h, \hw)$ is also regular and is smaller than $(t_f, \fu, t_g, \gv)$ in fashion (a), which means $(\bar{t}_0, f_0^{[\u_0]}, \bar{t}_h, \hw)$ lies in $Done$ and $\bar{t}_0(f_0^{[\u_0]})$ is gen-rewritable by $G_{end}$. Then $t_0(f_0^{[\u_0]})$ is also gen-rewritable by $G_{end}$, since $\bar{t}_0$ divides $t_0$. This is also contradicts with the property that $t_0(f_0^{[\u_0]})$ is {\em not} gen-rewritable by $G_{end}$.
\smallskip\\
In either case, the {\bf Second} fact is proved.

Proof of the {\bf Third} fact. According to the second fact, we have $\lpp(\brf) \succeq \lpp(t_0f_0)$ where $t_0 = \lpp(t_f\u)/\lpp(\u_0)$. Let $c_0 := \lc(\bru)/\lc(\u_0)$.
For the polynomial $\brfu - c_0 t_0 (f_0^{[\u_0]}) = (\brf - c_0 t_0 f_0)^{[\bru - c_0 t_0 \u_0]}$, we have $\lpp(\brf - c_0 t_0 f _0)\preceq \lpp(\brf)$ and $\lpp(\bru - c_0 t_0 \u_0)\prec \lpp(\bru)$. So $(\brf - c_0 t_0 f_0)^{[\bru - c_0 t_0 \u_0]}$ has a standard representation w.r.t. $G_{end}$ by {\bf Claim 1}. Note that $\lpp(\brf)\succeq \lpp(t_0 f _0)$ and $\lpp(\bru)=\lpp(t_0 \u_0)$. So after adding $c_0 t_0 f_0$ to both sides of the standard representation of $\brfu - c_0 t_0 (f_0^{[\u_0]})$, then we will get a standard representation of $\brfu$ w.r.t. $G_{end}$.
\end{proof}


\begin{proof}[Proof of {\bf Claim 3}]
Since $t_g(\gv)$ is gen-rewritable by $G_{end}$ and $\lpp(t_g\v) \prec \lpp(t_f\u)$, by using a similar method in the proof of the First and Second facts in {\bf Claim 2}, we have that there exists $g_0^{[\v_0]} \in G_{end}$ such that $t_g(\gv)$ is gen-rewritable by $g_0^{[\v_0]}$, $t_0(g_0^{[\v_0]})$ is not gen-rewritable by $G_{end}$ and $\lpp(t_g g) \succeq \lpp(t_0 g_0)$ where $t_0 = \lpp(t_g\v)/\lpp(\v_0)$.

If $\lpp(t_0 g_0) = \lpp(t_g g) = \lpp(t_f f)$, then the critical pair $[\fu, g_0^{[\v_0]}]$ has two possible forms.
\smallskip\\
{Form 1:} $[\fu, g_0^{[\v_0]}] = (t_f, \fu, t_0, g_0^{[\v_0]})$. Since $\lpp(t_f\u) \succ \lpp(t_g\v) = \lpp(t_0 \v_0)$, the critical pair $[\fu, g_0^{[\v_0]}]$ is regular and is smaller than $(t_f, \fu, t_g, \gv)$ in fashion (d), which means $[\fu, g_0^{[\v_0]}]$ lies in $Done$ and $t_f(\fu)$ is gen-rewritable by $G_{end}$.
\smallskip\\
{Form 2:} $[\fu, g_0^{[\v_0]}] = (\bar{t}_f, \fu, \bar{t}_0, g_0^{[\v_0]})$ where $\bar{t}_f$ divides $t_f$ and $\bar{t}_f \not= t_f$. Since $\lpp(t_f\u) \succ \lpp(t_g\v) = \lpp(t_0 \v_0)$, the critical pair $(\bar{t}_f, \fu, \bar{t}_0, g_0^{[\v_0]})$ is also regular and is smaller than $(t_f, \fu, t_g, \gv)$ in fashion (a), which means $(\bar{t}_f, \fu, \bar{t}_0, g_0^{[\v_0]})$ lies in $Done$ and $\bar{t}_f(\fu)$ is gen-rewritable by $G_{end}$. Then $t_f(\fu)$ is also gen-rewritable by $G_{end}$, since $\bar{t}_f$ divides $t_f$.

Otherwise, $\lpp(t_g g) \succ \lpp(t_0 g_0)$ holds. Let $c := \lc(\v)/\lc(\v_0)$. For the polynomial $t_g \gv - c t_0 (g_0^{[\v_0]}) = (t_g g - c t_0 g_0)^{[t_g \v - c t_0 \v_0]}$, we have $\lpp(t_g g - c t_0 g_0) = \lpp(t_g g)$ and $\lpp(t_g \v - c t_0 \v_0) \prec \lpp(t_g \v)$. Then $(t_g g - c t_0 g_0)^{[t_g \v - c t_0 \v_0]}$ has a standard representation w.r.t. $G_{end}$ by {\bf Claim 1}, and hence, there exists $\hw\in G_{end}$ such that $\lpp(h)$ divides $\lpp(t_g g - c t_0 g_0)=\lpp(t_g g)$ and $\lpp(t_h \w) \preceq \lpp(t_g \v - c t_0 \v_0) \prec \lpp(t_g \v)$ where $t_h = \lpp(t_g g)/\lpp(h)$. Note that $\lpp(t_h h) = \lpp(t_g g) = \lpp(t_f f)$. The critical pair of $[\fu, \hw]$ also has two possible forms.
\smallskip\\
{Form 1:} $[\fu, \hw] = (t_f, \fu, t_h, \hw)$. Since $\lpp(t_f\u) \succ \lpp(t_g\v) \succ \lpp(t_h \w)$, the critical pair $[\fu, \hw]$ is regular and is smaller than $(t_f, \fu, t_g, \gv)$ in fashion (c), which means $[\fu, \hw]$ lies in $Done$ and $t_f(\fu)$ is gen-rewritable by $G_{end}$.
\smallskip\\
{Form 2:} $[\fu, \hw] = (\bar{t}_f, \fu, \bar{t}_h, \hw)$ where $\bar{t}_f$ divides $t_f$ and $\bar{t}_f \not= t_f$. Since $\lpp(t_f\u) \succ \lpp(t_g\v) \succ \lpp(t_h \w)$, the critical pair $(\bar{t}_f, \fu, \bar{t}_h, \hw)$ is also regular and is smaller than $[\fu, \gv]$ in fashion (a), which means $(\bar{t}_f, \fu, \bar{t}_h, \hw)$ lies in $Done$ and $\bar{t}_f(\fu)$ is gen-rewritable by $G_{end}$. Then $t_f(\fu)$ is also gen-rewritable by $G_{end}$, since $\bar{t}_f$ divides $t_f$.

{\bf Claim 3} is proved.
\end{proof}

\begin{remark}
The proof of Theorem \ref{thm_main} also indicates that, {all regular or super regular critical pairs of $G_{end}$ are gen-rewritable by $G_{end}$}.
\end{remark}

\section{Developing New Criteria} \label{sec_newcri}

Based on the generalized criterion, to develop new criteria for signature-based algorithms, it suffices to choose appropriate admissible partial orders for the generalized criterion.

For example, we can develop a new criterion by using the following admissible partial order implied by GVW's criteria: for any  $\fu, \gv \in G$, we say $\gv < \fu$ if one of the following two conditions holds:
\begin{enumerate}

\item  $\lpp(t'g) < \lpp(tf)$, where $t'= \frac{\lcm(\lpp(\u), \lpp(\v))}{\lpp(\v)}$ and $t= \frac{\lcm(\lpp(\u), \lpp(\v))}{\lpp(\u)}$ such that $t(\fu)$ and $t'(\gv)$ have the same signature, i.e. $\lpp(t\u) = \lpp(t'\v)$.

\item $\lpp(t'g) = \lpp(tf)$ and $\gv$ is added to $G$ later than $\fu$.

\end{enumerate}
Recently, we notice Huang also considers a similar order in \citep{Huang10}. Applying this admissible partial order in the generalized criterion of algorithm AGC, we get a new algorithm (named by NEW). This algorithm can be regarded as an improved version of GVW.

To test the efficacy of the new criterion, we implemented the algorithm NEW on Singular (version 3-1-2), and use two strategies for selecting critical pairs.
\smallskip\\
Minimal {\bf S}ignature Strategy: $(t_f, \fu, t_g, \gv)$ is selected from {\sl CPairs} only if there does {\em not} exist another critical pair $({t}_{\brf}, \brfu, {t}_{\brg}, \brgv) \in \mbox{\sl CPairs}$ such that $\lpp(t_{\brf} \bru) \prec \lpp(t_f\u)$;
\smallskip\\
Minimal {\bf D}egree Strategy: $(t_f, \fu, t_g, \gv)$ is selected from {\sl CPairs} if there does {\em not} exist another critical pair $({t}_{\brf}, \brfu, {t}_{\brg}, \brgv) \in \mbox{\sl CPairs}$ such that $\deg(\lpp(t_{\brf} \brf)) \prec \deg(\lpp(t_f f))$.
\smallskip
The proofs in Section \ref{sec_proof} ensure the algorithm NEW is correct for both strategies.


In the following table, we use (s) and (d) to refer the two strategies respectively. The order $\prec_1$ is the Graded Reverse Lex order and $\prec_2$ is extended from $\prec_1$ in the following way: $x^\alpha\e_i \prec_2 x^\beta\e_j$, if either $\lpp(x^\alpha f_i) \prec_1 \lpp(x^\beta f_j)$, or  $\lpp(x^\alpha f_i) = \lpp(x^\beta f_j)$ and $i > j$. This order $\prec_2$ has also been used in \citep{Gao10b, SunWang10b}. The examples are selected from \citep{Gao10b} and the timings are obtained on Core i5 $4\times 2.8$ GHz with 4GB memory running Windows 7.

\begin{table}[!ht] \label{data}
\centering \caption{ $\#all.$: number of all critical pairs generated in the computation;  $\#red.$: number of critical pairs that are really reduced in the computation; $\#gen.$: number of non-zero generators in the \gr basis in the last iteration but before computing a reduced \gr basis. ``Katsura5 (22)" means there are 22 non-zero generators in the reduced \gr basis of Katsura5.}

\begin{tabular}{|c||c|c||c|c||c|c|} \hline
 & NEW(s) &  NEW(d) & NEW(s) &  NEW(d) & NEW(s) &  NEW(d)\\ \hline\hline

 & \multicolumn{2}{|c||}{Katsura5 (22)} & \multicolumn{2}{|c||}{Katsura6 (41)}  & \multicolumn{2}{|c|}{Katsura7 (74)} \\ \hline

$\#all.$ & 351 & 378 & 1035 & 1275  & 3160 & 3160\\ \hline

$\#red.$ & 39 & 40 & 73 & 78 & 121 & 121\\ \hline

$\#gen.$ & 27 & 28 & 46 & 51 & 80 & 80\\ \hline

time(sec.) & 1.400 & 1.195 & 7.865 & 5.650 & 38.750 & 29.950\\ \hline\hline

 & \multicolumn{2}{|c||}{Katsura8 (143)} & \multicolumn{2}{|c||}{Cyclic5 (20)} &  \multicolumn{2}{|c|}{Cyclic6 (45)} \\ \hline

$\#all.$ & 11325 & 11325 & 1128 & 2080 & 18528 &  299925 \\ \hline

$\#red.$ & 244 & 244 & 56 & 78 & 231 & 834 \\ \hline

$\#gen.$ & 151 & 151 & 48 & 65 & 193 & 775 \\ \hline

time(sec.) & 395.844 & 310.908 & 2.708 & 2.630 & 106.736 & 787.288 \\ \hline\hline

\hline\end{tabular}
\end{table}

From the above table, we can see that the new criterion can reject redundant critical pairs effectively. We also notice that the timings are influenced by the strategies of selecting critical pairs. For some examples, the algorithm with minimal signature strategy has better performance. The possible reason is that less critical pairs are generated by this strategy. For other examples, the algorithm with minimal degree strategy cost less time. The possible reason is that, although the algorithm with the minimal degree strategy usually generates more critical pairs, the critical pairs which are really needed to be reduced usually have lower degrees.

\section{Conclusions and Future works} \label{sec_conclusion}

Signature-based algorithms are a popular kind of algorithms for computing \gr basis. A generalized criterion for signature-based algorithms is proposed in this paper. Almost all existing criteria of signature-based algorithms can be specialized by the generalized criterion, and we show in detail how the generalized criterion specializes to F5 and GVW's criteria. We also proved that if the partial order is admissible, the generalized criterion is always correct no matter which computing order of the critical pairs is used. Since the generalized criterion can specialize to F5 and GVW's criteria, the proof in this paper also ensures the correctness of F5 and GVW for any computing order of critical pairs.

The significance of this generalized criterion is to describe which kind of criterion is correct in signature-based algorithms. Moreover, the generalized criterion also provides an effective approach to check and develop new criteria for signature-based algorithms, i.e., if a new criterion can be specialized from the generalized criterion by using an admissible partial order, it must be correct; when developing new criteria, it suffices to choose admissible partial orders in the generalized criterion. We also develop a new effective criterion in this paper. We believe that if the admissible partial order is in fact a total order, then the generalized criterion can reject almost all useless critical pairs. The proof of the claim will be included in future works.

Note that the generalized criterion is just one application of Key Fact in Section \ref{sec_mainideas}. We believe more results can be deduced from Key Fact as well. Related works will also be included in our future papers.

However, there are still some open problems.

\noindent
{\bf Problem 1:} Is the generalized criterion still correct if the partial order is not admissible? We do know some partial orders lead to wrong criteria. For example, consider the following partial order which is not admissible: for any $\fu, \gv \in G$, we say $\gv < \fu$, if $g=0$ and $f\not=0$; otherwise, $\gv$ is added to $G$ {\bf\em earlier} than $\fu$. This partial order leads to a wrong criterion. Because the polynomials $f_1^{[\e_1]}, \cdots, f_m^{[\e_m]}$ are added to $G$ earlier than others, so using this partial order, the generalized criterion will reject almost all critical pairs that are generated later, which definitely leads to a wrong output unless $\{f_1^{[\e_1]}, \cdots, f_m^{[\e_m]}\}$ itself is a labeled \gr basis.

\noindent
{\bf Problem 2:} Does the labeled \gr basis always exist for any ideal? Clearly, if the algorithm AGC terminates, then labeled \gr basis always exists. Note that GVW also computes a labeled \gr basis, and recently we learn by private communication about that Gao et al. have proved the termination of GVW, so in that sense the existence of labeled \gr basis has also been proved.

\noindent
{\bf Problem 3:} Does the algorithm AGC always terminate in finite steps? Since GVW has a special demand on the computing order of critical pairs, the proof for the termination of GVW cannot ensure the termination of the algorithm AGC. However, after testing many examples, we have not found a counterexample that AGC does not terminate.

%
%
%




\begin{thebibliography}{99}

\bibitem[Albrecht and Perry, 2010]{Albrecht10}
M. Albrecht and J. Perry. F4/5. Preprint, arXiv:1006.4933v2 [math.AC], 2010.

\bibitem[Arri and Perry, 2010]{Arri10}
A. Arri and J. Perry. The F5 criterion revised. Preprint, arXiv:1012.3664v3 [math.AC], 2010.

\bibitem[Buchberger, 1979]{Buchberger79}
B. Buchberger. A criterion for detecting unnecessary reductions in the construction of \gr basis. In Proceedings of EUROSAM'79, Lect. Notes in Comp. Sci., Springer, Berlin, vol. 72, 3-21, 1979.

\bibitem[Buchberger, 1985]{Buchberger85}
B. Buchberger. \gr-bases: an algorithmic method in polynomial ideal theory. Reidel Publishing Company, Dodrecht - Boston - Lancaster, 1985.

\bibitem[Courtois et al., 2000]{Courtois00}
N. Courtois, A. Klimov, J. Patarin, and A. Shamir. Efficient algorithms for solving overdefined systems of multivariate polynomial equations. In Proceedings of EUROCRYPT'00, Lect. Notes in Comp. Sci., Springer, Berlin, vol. 1807, 392-407, 2000.

\bibitem[Cox et al., 2004]{CLO04} D. Cox, J. Little, and D. O'Shea. Using
algebraic geometry. Springer, New York, second edition, 2005.

\bibitem[Ding et al., 2008]{Ding08}
J. Ding, J. Buchmann, M.S.E. Mohamed, W.S.A.E. Mohamed, and R.-P. Weinmann. MutantXL. In Proceedings of the 1st international conference on Symbolic Computation and Cryptography (SCC08), Beijing, China, 16-22, 2008.

\bibitem[Eder, 2008]{Eder08}
C. Eder. On the criteria of the F5 algorithm. Preprint, arXiv:0804.2033v4 [math.AC], 2008.

\bibitem[Eder and Perry, 2010]{Eder09}
C. Eder and J. Perry. F5C: a variant of Faug\`ere's F5 algorithm with reduced \gr bases. J. Symb. Comput., vol. 45(12), 1442-1458, 2010.

\bibitem[Eder and Perry, 2011]{Eder11}
C. Eder and J. Perry. Signature-based Algorithms to Compute Gr\"obner Bases. In Proceedings of ISSAC'11, ACM Press, New York, USA, 99-106, 2011.


\bibitem[Faug\`ere, 1999]{Fau99}
J.-C. Faug\`ere. A new effcient algorithm for computing Gr\"obner bases ($F_4$). J. Pure Appl. Algebra, vol. 139(1-3), 61-88, 1999.

\bibitem[Faug\`ere, 2002]{Fau02}
J.-C. Faug\`ere. A new effcient algorithm for computing Gr\"obner bases without reduction to zero ($F_5$). In Proceedings of ISSAC'02, ACM Press, New York, USA, 75-82, 2002. Revised version downloaded from fgbrs.lip6.fr/jcf/Publications/index.html.

\bibitem[Gao et al., 2010a]{Gao09}
S.H. Gao, Y.H. Guan, and F. Volny. A new incremental
algorithm for computing \gr bases. In Proceedings of ISSAC'10, ACM Press, New York, USA, 13-19, 2010.

\bibitem[Gao et al., 2010b]{Gao10b}
S.H. Gao, F. Volny, and M.S. Wang. A new algorithm for computing \gr bases. Cryptology ePrint Archive, Report 2010/641, 2010.

\bibitem[Gebauer and Moller, 1986]{GebMol86}
R. Gebauer and H.M. Moller. Buchberger's algorithm and staggered linear bases. In Proceedings of SYMSAC'86, ACM press, New York, USA, 218-221, 1986.

\bibitem[Giovini et al., 1991]{Gio91}
A. Giovini, T. Mora, G. Niesi, L. Robbiano and C. Traverso. ``One sugar cube, please" or selection strategies in the Buchberger algorithm. In Proceedings of ISSAC'91, ACM Press, New York, USA, 49-54, 1991.

\bibitem[Hashemi and Ars, 2010]{Ars09}
A. Hashemi and G. Ars. Extended F5 criteria. J. Symb. Comput., vol. 45(12), 1330-1340, 2010.

\bibitem[Huang, 2010]{Huang10}
L. Huang. A new conception for computing Gr\"obner basis and its applications. Preprint, arXiv:1012.5425v2 [cs.SC], 2010.

\bibitem[Lazard, 1983]{Lazard83}
D. Lazard. Gr\"obner bases, Gaussian elimination and resolution of systems
of algebraic equations. In Proceeding of EUROCAL'83, Lect. Notes in Comp.
Sci., Springer, Berlin, vol. 162, 146-156, 1983.

\bibitem[M\"oller et al., 1992]{Mora92}
H.M. M\"oller, T. Mora, and C. Traverso. Gr\"obner bases computation using syzygies. In Proceedings of ISSAC'92, ACM Press, New York, USA, 320-328, 1992.

\bibitem[Stegers, 2006]{Stegers05}
T. Stegers. Faug\`ere's F5 algorithm revisited. Cryptology ePrint Archive, Report
2006/404, 2006.

\bibitem[Sun and Wang, 2010a]{SunWang10a}
Y. Sun and D.K. Wang. The F5 algorithm in Buchberger's style. To appear in J. Syst. Sci. Complex., arXiv:1006.5299v2 [cs.SC], 2010.

\bibitem[Sun and Wang, 2010b]{SunWang10b}
Y. Sun and D.K. Wang. A new proof for the correctness of the F5 algorithm. Preprint, arXiv:1004.0084v4 [cs.SC], 2010.

\bibitem[Sun and Wang, 2011]{SunWang11}
Y. Sun and D.K. Wang. A Generalized Criterion for  Signature Related \gr Basis Algorithms. In Proceedings of ISSAC'11, ACM Press, New York, USA, 337-344, 2011.



\bibitem[Zobnin, 2010]{Zobnin10}
A. Zobnin. Generalization of the F5 algorithm for calculating Gr\"obner bases for polynomial ideals. Programming and Computer Software, vol. 36(2), 75-82, 2010.

\end{thebibliography}
\end{document}